\documentclass{article}
\usepackage{amsmath, amssymb, amsthm}
\usepackage{fullpage}
\usepackage{verbatim}
\usepackage{microtype}
\usepackage[margin=1.1in]{geometry}
\usepackage{times}
\usepackage[utf8]{inputenc}
\title{Computing Marginals Using MapReduce}
\author{Foto Afrati\thanks{NTU Athens} \and
Shantanu Sharma\thanks{Ben Gurion University} \and
Jeffrey D.~Ullman\thanks{Stanford University} \and
Jonathan R.~Ullman\thanks{Northeastern University}}
\date{\today}

\newtheorem{theorem}{Theorem}[section]

\newtheorem{lemma}[theorem]{Lemma}

\newtheorem{corollary}[theorem]{Corollary}

\newtheorem{example}[theorem]{Example}
\theoremstyle{definition}

\begin{document}

\maketitle
\begin{abstract}
We consider the problem of computing the data-cube marg\-i\-nals of a fixed order $k$ (i.e., all marginals that aggregate over $k$ dimensions), using a single round of MapReduce. The focus is on the relationship between the reducer size (number of inputs allowed at a single reducer) and the replication rate (number of reducers to which an input is sent). We show that the replication rate is minimized when the reducers receive all the inputs necessary to compute one marginal of higher order. That observation lets us view the problem as one of covering sets of $k$ dimensions with sets of a larger size $m$, a problem that has been studied under the name ``covering numbers.'' We offer a number of constructions that, for different values of $k$ and $m$ meet or come close to yielding the minimum possible replication rate for a given reducer size.
\end{abstract}

\section{Background}
\label{background-sect}

\subsection{Marginals}
\label{marginals-sect}
Consider an $n$-dimensional data cube \cite{GBLP96} and the computation of its marginals by MapReduce. A {\em marginal} of a data cube is the aggregation of the data in all those tuples that have fixed values in a subset of the dimensions of the cube. We shall assume this aggregation is the sum, but the exact nature of the aggregation is unimportant in what follows. Marginals can be represented by a list whose elements correspond to each dimension, in order. If the value in a dimension is fixed, then the fixed value represents the dimension. If the dimension is aggregated, then there is a * for that dimension. The number of dimensions over which we aggregate is the {\em order} of the marginal.

\begin{example}
\label{marg-ex}
Suppose $n=5$, and the data cube is a relation DataCube(D1,D2,D3,D4,D5,V). Here, D1 through D5 are the dimensions, and $V$ is the value that is aggregated.
\begin{verbatim}
SELECT SUM(V)
FROM DataCube
WHERE D1 = 10 AND D3 = 20 AND D4 = 30;
\end{verbatim}
will sum the data values in all those tuples that have value 10 in the first dimension, 20 in the third dimension, 30 in the fourth dimension, and any values in the second and fifth dimension of a five-dimensional data cube. We can represent this marginal by the list $[10,*,20,30,*]$, and it is a second-order marginal.
\end{example}

\subsection{Assumption: All Dimensions Have Equal Extent}
\label{equal-d-subsect}
We shall make the simplifying assumption that in each dimension there are $d$ different values. In practice, we do not expect to find that each dimension really has the same number of values. For example, if one dimension represents Amazon customers, there would be millions of values in this dimension. If another dimension represents the date on which a purchase was made, there would ``only'' be thousands of different values.

However, it probably makes little sense to compute marginals where we fix the Customer dimension to be each customer, in turn; there would be too many marginals, and each would have only a small significance. More likely, we would want to group the values of dimensions in some way, e.g., customers by state and dates by month. Moreover, we shall see that our methods really only need the parameter $d$ to be an upper bound on the true number of distinct values in a dimension. The consequence of the {\em extents} (number of distinct values) of different dimensions being different is that some of the reducers will get fewer than the theoretical maximum number of inputs allowed. That discrepancy has only a small effect on the performance of the algorithm. Moreover, if there are really large differences among the extents of the dimensions, then an extension of our algorithms can improve the performance. We shall defer this issue to Section~\ref{different-sizes-sect}.

\subsection{Mapping Schemas for MapReduce Algorithms}
\label{mapping-schema-subsect}
We assume the reader is familiar with the MapReduce computational model \cite{dean:04}. Following the approach to analyzing MapReduce algorithms given in \cite{ADSU13}, we look at tradeoffs between the {\em reducer size} (maximum number of inputs allowed at a reducer), which we always denote by $q$, and the {\em replication rate} (average number of reducers to which an input needs to be sent), which we always denote by $r$. The replication rate represents the cost of communication between the mappers and reducers, and communication cost is often the dominant cost of a MapReduce computation. Typically, the larger the reducer size, the lower the replication rate. But we want to keep the reducer size low for two reasons: it enables computation to proceed in main memory, and it forces a large degree of parallelism, both of which lead to low wall-clock time to finish the MapReduce job.

In the theory of \cite{ADSU13}, a problem is modeled by a set of inputs (the tuples or points of the data cube, here), a set of outputs (the values of the marginals) and a relationship between the inputs and outputs that indicates which inputs are needed to compute which outputs. In order for an algorithm to solve this problem with a reducer size $q$, there must be a {\em mapping schema}, which is a relationship between inputs and reducers that satisfies two properties:

\begin{enumerate}
\item
No reducer is associated with more than $q$ inputs, and
\item
For every output, there is some reducer that is associated with all the inputs that output needs for its computation.
\end{enumerate}

Point (1) is the definition of ``reducer size,'' while point~(2) is the requirement that the algorithm can compute all the outputs that the problem requires. The fundamental reason that MapReduce algorithms are not just parallel algorithms in general is embodied by point~(2). In a MapReduce computation, every output is computed by one reducer, independently of all other reducers.

\subsection{Na\"ive Solution: Computing One Marginal Per Reducer}
\label{min-q-subsect}
Now, let us consider the problem of computing all the marginals of a data cube in the above model. If we are not careful, the problem becomes trivial. The marginal that aggregates over all dimensions is an output that requires all $d^n$ inputs of the data cube. Thus, $q=d^n$ is necessary to compute all the marginals. But that means we need a single reducer as large as the entire data cube, if we are to compute all marginals in one round. As a result, it only makes sense to consider the problem of computing a limited set of marginals in one round.

The $k^{\mathrm{th}}$-order marginals are those that fix $n-k$ dimensions and aggregate over the remaining $k$ dimensions. To compute a $k^{\mathrm{th}}$-order marginal, we need $q\ge d^k$, since such a marginal aggregates over $d^k$ tuples of the cube. Thus, we could compute all the $k^{\mathrm{th}}$-order marginals with $q=d^k$, using one reducer for each marginal. As a ``problem'' in the sense of \cite{ADSU13}, there are $d^n$ inputs, and $d^{n-k}\binom{n}{k}$ outputs, each representing one of the marginals. Each output is connected to the $d^k$ inputs over which it aggregates. Each input contributes to $\binom{n}{k}$ marginals~-- those marginals that fix $n-k$ out of the $n$ dimensions in a way that agrees with the tuple in question. That is, for $q=d^k$, we can compute all the $k^{\mathrm{th}}$-order marginals with a replication rate $r$ equal to $\binom{n}{k}$.

For $q=d^k$, there is nothing better we can do. However, when $q$ is larger, we have a number of options, and the purpose of this paper is to explore these options.

\section{Related Work}
\label{related-work-sect}
There have been a number of papers that look at the problem of using MapReduce to compute marginals. Probably the closest work to what we present here is in \cite{RP14}. This paper expresses the goal of minimizing communication, and of partitioning the work among reducers. It does not, however, present concrete bounds or algorithms that meet or approach those bounds, as we shall do here.

\cite{NYBR12} considers constructing a data cube using a nonassociative aggregation function and also examines how to deal with nonuniformity in the density of tuples in the cube. Like all the other papers mentioned, it deals with constructing the entire data cube using multiple rounds of MapReduce. We consider how to compute only the marginals of one order, using one round. We may assume that locally, at each reducer, higher-order marginals are computed by aggregating lower-order marginals for efficiency, but this method does not result in additional MapReduce rounds.

\cite{AFR11} looks at using MapReduce to form a data cube from data stored in Bigtable. \cite{Lee12} and \cite{WGRO14} are implementations of known algorithms in MapReduce. Finally, \cite{WCTAAX13} talks about extending MapReduce to compute data cubes more efficiently.

\section{Computing Many Marginals at One Reducer}
\label{higher-q-sect}
We wish to study the tradeoff between reducer size and replication rate, a phenomenon that appears in many problems \cite{Ull12, AU13, ADSU13, FDSU15}. Since we know from Section~\ref{min-q-subsect} the minimum possible reducer size, we need to consider whether using a larger value of $q$ can result in a significantly smaller value of $r$. Our goal is to combine marginals, in such a way that there is maximum overlap among the inputs needed to compute marginals at the same reducer.

We shall start by assuming the maximum overlap and minimum replication rate is obtained by combining $k^{\mathrm{th}}$-order marginals into a set of inputs suitable to compute one marg\-i\-nal of order higher than $k$. This assumption is correct, and we shall offer a proof of the fact in Section~\ref{isoper-sect}. However, we are still left with answering the question: how do we pack all the $k^{\mathrm{th}}$-order marginals into as few marginals of order $m$ as possible, for each $m>k$. That will lead us to the matter of ``asymmetric covering codes'' or ``covering numbers'' \cite{applegate:asym, cooper:asym}.

\subsection{Covering Marginals}
\label{covering-marg-subsect}
Suppose we want to compute all $k^{\mathrm{th}}$-order marginals, but we are willing to use reducers of size $q=d^m$ for some $m>k$. If we fix any $n-m$ of the $n$ dimensions of the data cube, we can send to one reducer the $d^m$ tuples of the cube that agree with those fixed values. We then can compute all the marginals that have $n-k$ fixed values, as long as those values agree with the $n-m$ fixed values that we chose originally.

\begin{example}
\label{covering-ex}
Let $n=7$, $k=2$, and $m=3$. Suppose we fix the first $n-m = 4$ dimensions, say using values $a_1$, $a_2$, $a_3$, and $a_4$. Then we can cover the $d$ marginals $a_1a_2a_3a_4x**$ for any of the values $x$ that may appear in the fifth dimension. We can also cover all marginals $a_1a_2a_3a_4{*}y*$ and $a_1a_2a_3a_4{*}{*}z$, where $y$ and $z$ are any of the possible values for the sixth and seventh dimensions, respectively. Thus, we can cover a total of $3d$ second-order marginals at this one reducer. That turns out to be the largest number of marginals we can cover with one reducer of size $q=d^3$.
\end{example}

What we do for one assignment of $n-m$ values to $n-m$ of the dimensions we can do for all assignments of values to the same dimensions, thus creating $d^{n-m}$ reducers, each of size $d^m$. Together, these reducers allow us to compute all $k^{\mathrm{th}}$-order marginals that fix the same $n-m$ dimensions (along with any $m-k$ of the remaining $m$ dimensions).

\begin{example}
\label{covering2-ex}
Continuing Example~\ref{covering-ex}, if we use $d^4$ reducers, each of which fixes the first four dimensions, then we can cover any of the $3d^5$ marginals that fix the first four dimensions along with one of dimensions 5, 6, or 7.
\end{example}

\subsection{From Marginals to Sets of Dimensions}
\label{marg-handle-subsect}
To understand why the problem is more complex than it might appear at first glance, let us continue thinking about the simple case of Example~\ref{covering-ex}. We need to cover all second-order marginals, not just those that fix the first four dimensions. If we had one team of $d^4$ reducers to cover each four of the seven dimensions, then we would surely cover all second-order marginals. But we don't need all $\binom{7}{2} = 21$ such teams. Rather, it is sufficient to pick a collection of sets of four of the seven dimensions, such that every set of five of the seven dimensions contains one of those sets of size four.

In what follows, we find it easier to think about the sets of dimensions that are aggregated, rather than those that are fixed. So we can express the situation above as follows. Collections of second-order marginals are represented by pairs of dimensions~-- the two dimensions such that each marginal in the collection aggregates over those two dimensions. These pairs of dimensions must be {\em covered} by sets of three dimensions~-- the three dimensions aggregated over by one third-order marginal. Our goal, which we shall realize in Example~\ref{7-5-1-ex} below, is to find a smallest set of tripletons such that every pair chosen from seven elements is contained in one of those tripletons.

In general, we are faced with the problem of covering all sets of $k$ out of $n$ elements by the smallest possible number of sets of size $m>k$. Such a solution leads to a way to compute all $k^{\mathrm{th}}$-order marginals using as few reducers of size $d^m$ as possible. Abusing the notation, we shall refer to the sets of $k$ dimensions as {\em marginals}, even though they really represent teams of reducers that compute large collections of marginals with the same fixed dimensions. We shall call the larger sets of size $m$ {\em handles}. The implied MapReduce algorithm takes each handle and creates from it a team of reducers that are associated, in all possible ways, with fixed values in all dimensions except for those dimensions in the handle. Each created reducer receives all inputs that match its associated values in the fixed dimensions.

\begin{example}
\label{7-5-1-ex}
Call the seven dimensions $ABCDEFG$. Then here is a set of seven handles (sets of size three), such that every marginal of size two is contained in one of them:
$$ABC,~ADE,~AFG,~BDF,~BEG,~CDG,~CEF$$
To see why these seven handles suffice, consider three cases, depending on how many of $A$, $B$, and $C$ are in the pair of dimensions to be covered.

\smallskip

\noindent Case 0: If none of $A$, $B$ or $C$ is in the marginal, then the marginal consists of two of $D$, $E$, $F$, and $G$. Note that all six such pairs are contained in one of the last six of the handles.

\smallskip

\noindent Case 1: If one of $A$, $B$, or $C$ is present, then the other member of the marginal is one of $D$, $E$, $F$, or $G$. If $A$ is present, then the second and third handles, $ADE$ and $AFG$ together pair $A$ with each of the latter four dimensions, so the marginal is covered. If $B$ is present, a similar argument involving the fourth and fifth of the handles suffices, and if $C$ is present, we argue from the last two handles.

\smallskip

\noindent Case 2: If the marginal has two of $A$, $B$, and $C$, then the first handle covers the marginal.
\end{example}

Incidentally, we cannot do better than Example~\ref{7-5-1-ex}. Since no handle of size three can cover more than three marginals of size two, and there are $\binom{7}{2}=21$ marginals, clearly seven handles are needed.

As a strategy for evaluating all second-order marginals of a seven-dimensional cube, let us see how the reducer size and replication rate compare with the baseline of using one reducer per marginal. Recall that if we use one reducer per marginal, we have $q=d^2$ and $r=\binom{7}{5}=21$. For the present method, we have $q=d^3$ and $r=7$. That is, each tuple is sent to the seven reducers that have the matching values in dimensions $DEFG$, $BCFG$, and so on, each set of attributes on which we match corresponding to the complement of one of the seven handles mentioned in Example~\ref{7-5-1-ex}.

\subsection{Covering Numbers}
\label{covering-numbers-subsect}
Let us define $C(n,m,k)$ to be the minimum number of sets of size $m$ out of $n$ elements such that every set of $k$ out of the same $n$ elements is contained in one of the sets of size $m$. For instance, Example~\ref{7-5-1-ex} showed that $C(7,3,2)=7$. $C(n,m,k)$ is called the {\em covering number} in \cite{applegate:asym}. The numbers $C(n,m,k)$ guide our design of algorithms to compute $k^{\mathrm{th}}$-order marginals.  There is an important relationship between covering numbers and replication rate, that justifies our focus on constructive upper bounds for $C(n,m,k)$.

\begin{theorem}
\label{rr-th}
If $q=d^{m}$, then we can solve the problem of computing all $k^{\mathrm{th}}$-order marginals of an $n$-dimensional data cube with $r = C(n,m,k)$.
\end{theorem}

\begin{proof}
Each marginal in the set of $C(n,m,k)$ handles can be turned into a team of reducers, one for each of the $d^{n-m}$ ways to fix the dimensions that are not in the handle. Each input gets sent to exactly one member of the team for each handle~-- the reducer that corresponds to fixed values that agree with the input. Thus, each input is sent to exactly $C(n,m,k)$ reducers.
\end{proof}

Sometimes we will want to fix some choices of $m, k$ and study how $C(n,m,k)$ grows with the dimension $n$.  In this case we will often write simply $C(n)$ when $m$ and $k$ are clear from context.

\subsection{First-Order Marginals}
\label{first-order-subsect}
The case $k=1$ is quite easy to analyze. We are asking how many sets of size $m$ are needed to cover each singleton set, where the elements are chosen from a set of size $n$. It is easy to see that we can group the $n$ elements into $\lceil n/m\rceil$ sets so that each of the $n$ elements is in at least one of the sets, and there is no way to cover all the singletons with fewer than this number of sets of size $m$. That is, $C(n,m,1) = \lceil n/m\rceil$. For example, If $n=7$ and $m=2$, then the seven dimensions $ABCDEFG$ can be covered by four sets of size 2, such as $AB$, $CD$, $EF$, and $FG$.

\subsection{2nd-Order Marginals Covered by 3rd-Or\-der Handles}
\label{2-3-subsect}
The next simplest case is $C(n,3,2)$, that is, covering second-order marginals by third-order marginals, or equivalently, covering sets of two out of $n$ elements by sets of size 3. One simple observation is that a set of size 3 can cover only three pairs, so $C(n,3,2)\ge \binom{n}{2}/3$, or:
\begin{equation}
\label{Cn32-lower-eq}
C(n,3,2)\ge n^2/6 - n/6
\end{equation}
In fact, more generally, $C(n,m,k)\ge \binom{n}{k}/\binom{m}{k}$.

\noindent{\bf Aside}: We also remark that, using the probabilistic method, one can show that
$$C(n,m,k) \leq 2 \ln \binom{n}{k} \cdot \frac{\binom{n}{k}}{\binom{m}{k}}$$
so this simple lower bound is actual optimal up to a factor of $2 \ln \binom{n}{k}$.  However, in what follows we will give constructions that are

\begin{itemize}
\item[(a)]
 \emph{Explicit} and, more importantly,
\item[(b)]
Meet the lower bound either \emph{exactly} or to within a constant factor.
\end{itemize}

While \cite{applegate:asym} gives us some specific optimal values of $C(n,3,2)$ to use as the basis of an induction, we would like a recursive algorithm for constructing ways to cover sets of size 2 by sets of size 3, and we would like this recursion to yield solutions that are as close to the lower bound of Equation~\ref{Cn32-lower-eq} as possible. We can in fact give a construction that, for an infinite number of $n$, matches the lower bound of Equation~\ref{Cn32-lower-eq}.
Suppose we have a solution for $n$ dimensions. We construct a solution for $3n$ dimensions as follows. First, group the $3n$ dimensions into three groups of $n$ each. Let these groups be $\{A_1,A_2,\ldots,A_n\}$, $\{B_1,B_2,\ldots,B_n\}$, and $\{C_1,C_2,\ldots,C_n\}$. We construct handles of two kinds:

\begin{enumerate}
\item
Choose all sets of three elements, one from each group, say $A_iB_jC_k$, such that $i+j+k$ is divisible by $n$. There are evidently $n^2$ such handles, since any choice from the first two groups can be completed by exactly one choice from the third group.
\item
Use the assumed solution for $n$ dimensions to cover all the pairs chosen from one of the three groups. So doing adds another $3C(n,3,2)$ handles.
\end{enumerate}

This set of handles covers all pairs chosen from the $3n$ dimensions. In proof, if the pair has dimensions from different groups, then it is covered by the handle from (1) that has those two dimensions plus the unique member of the third group such that the sum of the three indexes is divisible by $n$. If the pair comes from a single group, then we can argue recursively that it is covered by a handle added in (2).

\begin{example}
\label{3-2-ex}
Let $n=3$, and let the three groups be $A_1A_2A_3$, $B_1B_2B_3$, and $C_1C_2C_3$. From the first rule, we get the handles $A_1B_1C_1$, $A_1B_2C_3$, $A_1B_3C_2$, $A_2B_1C_3$, $A_2B_2C_2$, $A_2B_3C_1$, $A_3B_1C_2$, $A_3B_2C_1$, and $A_3B_3C_3$. Notice that the sum of the subscripts in each handle is 3, 6, or 9.
For the second rule, note that when $n=3$, a single handle consisting of all the dimensions suffices. Thus, we need to add $A_1A_2A_3$, $B_1B_2B_3$, and $C_1C_2C_3$. the total number of handles is 12. This set of handles is as small as possible, since $\binom{9}{2}/3 = 12$.
\end{example}

The recurrence that results from this construction is:
\begin{equation}
\label{Cn32-recur-eq}
C(3n,3,2)\le n^2 + 3C(n,3,2)
\end{equation}
Let us use $C(n)$ as shorthand for $C(n,3,2)$ in what follows. We claim that;

\begin{theorem}
\label{Cn32-th}
For $n$ a power of 3:
$C(n) = n^2/6 - n/6$.
\end{theorem}

\begin{proof}
We already argued that $C(n)\ge n^2/6 - n/6$, so we have only to show $C(n)\le n^2/6 - n/6$ for $n$ a power of 3.

For the basis, $C(3)=1$. Obviously one set of the three elements covers all three of its subsets of size two. Since $1 = 3^2/6 - 3/6$, the basis is proven.

For the induction, assume $C(n)\le n^2/6 - n/6$. Then by Equation~\ref{Cn32-recur-eq}, $C(3n)\le n^2 + 3n^2/6 - 3n/6 = 3n^2/2 - n/2 = (3n)^2/6 - (3n)/6$.
\end{proof}

We can get the same bound, or close to the same bound, for values of $n$ that are not a power of 3 if we start with another basis. All optimal values of $C(n)$ up to $n=13$ are given in \cite{applegate:asym}. For $n=4,5,\ldots,13$, the values of $C(n)$ are 3, 4, 6, 7, 11, 12, 17, 19, 24, and 26.

Using Theorem~\ref{rr-th}, we have the following corollary to Theorem~\ref{Cn32-th}.

\begin{corollary}
\label{Cn32-corr}
If $q=d^3$ and $n$ is a power of 3, then we can compute all second-order marginals with a replication rate of $n^2/6 - n/6$.
\end{corollary}

Note that the bound on replication rate given by Corollary~\ref{Cn32-corr}, which is equivalent to $\binom{n}{2}/3$, is exactly one third of the replication rate that would be necessary if we used a single reducer for each marginal (or, since $q=d^3$, and one second-order marginal requires $d^2$ inputs, the same improvement would hold when compared with packing $d$ randomly chosen second-order marginals at each reducer).

\subsection{A Slower Recursion for 2nd-Order Marg\-inals}
\label{2-3-slow-subsect}
There is an alternative recursion for constructing handles that offers solutions for $C(n,3,2)$. This recursion is not as good asymptotically as that of Section~\ref{2-3-subsect}; it uses approximately $n^2/4$ rather than $n^2/6$ handles. However, this recursion gives solutions for any $n$, not just those that are powers of 3.

Note that if we attempt to address values of $n$ that are not a power of $3$ by simply rounding $n$ up to the nearest power of $3$ and using the recursive construction from the previous section, then we may increase the replication rate by a factor as large as $9$, whereas the recursion in this section is never suboptimal by a factor larger than $3/2$.

Let us call the $n$ dimensions $A_1A_2B_1B_2\cdots B_{n-2}$.  We choose handles of two kinds:

\begin{enumerate}
\item
Handles that contain $A_1$, $A_2$, and one of
$$B_1,B_2,\ldots,B_{n-2}$$
There are clearly $n-2$ handles of this kind.
\item
The $C(n-2)$ handles that recursively cover all pairs chosen from $B_1,B_2,\ldots, B_{n-2}$.
\end{enumerate}

We claim that every marginal of size 2 is covered by one of these handles. If the marginal has neither $A_1$ nor $A_2$, then clearly it is covered by one of the handles from (2). If the marginal has both $A_1$ and $A_2$, then it is covered by any of the handles from (1). And if the marginal has one but not both of $A_1$ and $A_2$, then it has exactly one of the $B_i$'s. Therefore, it is covered by the handle from (1) that has $A_1$, $A_2$, and that $B_i$.

\begin{example}
\label{3-2-slow-ex}
Let $n=6$, and call the dimensions
$$ABCDEF$$
where $A$ and $B$ form the first group, and $CDEF$ form the second group. By rule~(1), we include handles $ABC$, $ABD$, $ABE$, and $ABF$. By rule~(2) we have to add a cover for each pair from $CDEF$. One choice is $CDE$, $CDF$, and $DEF$, for a total of seven handles. This choice is not exactly optimal, since six handles of size three suffice to cover all pairs chosen from six elements~\cite{applegate:asym}.
\end{example}

The resulting recurrence is
$$C(n)\le n-2+C(n-2)$$
We claim that for odd $n\ge3$, $C(n)\le n^2/4 - n/2 +1/4$. For the basis, we know that $C(3)=1$. As $3^2/4 - 3/2 + 1/4 = 1$, the basis $n=3$ is proved. The induction then follows from the fact that
$$n-2 +(n-2)^2/4 -(n-2)/2 + 1/4 = n^2/4 -n/2 + 1/4$$
For even $n$, we could start with $C(4)=3$. But we do slightly better if we start with the value $C(6)=6$, given in \cite{applegate:asym}. That gives us $C(n)\le n^2/4-n/2$ for all even $n\ge6$.

While this recurrence gives values of $C(n)$ that grow with $n^2/4$ rather than $n^2/6$, it does give us values that the recurrence of Section~\ref{2-3-subsect} cannot give us.

\begin{example}
\label{combine-ex}
The recurrence of Section~\ref{2-3-subsect} gives us
$$C(27) = 117$$
If we want a result for $n=31$, we can apply the recurrence of this section twice, to get $C(29) \le 27+117 = 143$ and $C(31) \le 29+143 = 172$. In comparison, the lower bound on the number of handles needed for $n=31$ is $\binom{31}{2}/3 = 155$.
\end{example}

\subsection{Aside: Solving Recurrences}
\label{solve-recur-subsect}
We are going to propose several recurrences that describe inductive constructions of sets of handles. While we do not want to explain how one discovers the solution to each recurrence, there is a general pattern that can be used by the reader who wants to see how the solutions are derived; see \cite{AU95}.

A recurrence like $C(n)\le n-2+C(n-2)$ from Section~\ref{2-3-slow-subsect} will have a solution that is a quadratic polynomial, say $C(n)=an^2+bn+c$. It turns out that the constant term $c$ is needed only to make the basis hold, but we can get the values of $a$ and $b$ by replacing the inequality by an equality, and then recognizing that the terms depending on $n$ must be 0. In this case, we get
$$an^2+bn+c = n-2 +a(n-2)^2+b(n-2)+c$$
or
$$an^2+bn+c = n-2 +an^2 -4an +4a +bn -2b +c$$
Cancelling terms and bringing the terms with $n$ to the left, we get
$$n(4a-1) = 4a - 2b -2$$
Since a function of $n$ cannot be a constant unless the coefficient of $n$ is 0, we know that $4a-1=0$, or $a=1/4$. The right side of the equality must also be 0, so we get $4(1/4)-2b-2=0$, or $b=-1/2$. We thus know that $C(n) = n^2/4-n/2+c$ for some constant $c$, depending on the basis value.

\subsection{Covering 2nd-Order Marginals With Larg\-er Handles}
\label{2-m-subsect}
We can view the construction of Section~\ref{2-3-slow-subsect} as dividing the dimensions into two groups; the first consisted of only $A_1$ and $A_2$, while the second group consisted of the remaining dimensions, which we called $B_1,B_2,\ldots,B_{n-2}$. We then divided the second-order marginals, which are pairs of dimensions, according to how the pair was divided between the groups. That is, either 0, 1, or 2 of the dimensions could be the the first group $\{A_1,A_2\}$. We treated each of these three cases, as we can summarize in the table of Fig.~\ref{cases-fig}.

\begin{figure}[htfb]

\begin{center}
\begin{tabular}{c|l l}
Case & $\{A_1,A_2\}$ & ~~~$B_i$'s\\
\hline
0 & none & ~~cover\\
1 & ~~~~~~~not & needed\\
2 & $A_1A_2$ & ~~all $B_i$'s\\
\end{tabular}\end{center}

\caption{How we cover each of the three cases: 0, 1, or 2 dimensions of the marginal are in the first group ($A_1A_2$)}
\label{cases-fig}
\end{figure}

That is, marginals with zero of $A_1$ and $A_2$ (Case 0) are covered recursively by the best possible set of handles that cover the $B_i$'s. Marginals with both $A_1$ and $A_2$ (Case 2) are covered by many handles, since we add to $A_1A_2$ all possible sets of size 1 formed from the $B_i$'s. The reason we do so is that we can then cover all the marginals belonging to Case~1, where exactly one of $A_1$ and $A_2$ is present, without adding any additional handles. That is, had we been parsimonious in Case~2, and only included one handle, such as $A_1A_2B_1$, then we would not have been able to skip Case~1.

Now, let us turn our attention to covering pairs of dimensions by sets of size larger than three; i.e., we wish to cover second-order marginals by handles of size $m$, for some $m\ge4$. We can generalize the technique of Section~\ref{2-3-slow-subsect} by using one group of size $m-1$, say $A_1,A_2,\ldots,A_{m-1}$ and another group with the remaining dimensions, $B_1,B_2,\ldots,B_{n-(m-1)}$. We can form handles for Case~0, where none of the $A_i$'s are in the marginal, recursively as we did in Section~\ref{2-3-slow-subsect}. That requires $C(n-(m-1),m,2)$ handles. If we deal with Case~$m-1$ by adding to $A_1A_2\cdots A_{m-1}$ each of the $B_i$'s in turn, to form $n-(m-1)$ additional handles, we cover all the other cases. Of course all the cases except for Case~1, where exactly one of the $A_i$'s is in the marginal, are vacuous. This reasoning gives us a recurrence:
$$C(n,m,2)\le n-(m-1) + C(n-(m-1),m,2)$$

Using the technique suggested by Section~\ref{solve-recur-subsect}. along with the obvious basis case $C(m,m,2) = 1$, we get the solution:
$$C(n,m,2)\le \frac{n^2}{2(m-1)} - \frac{n}{2} +1 - \frac{m}{2(m-1)}$$
Note that asymptotically, this solution uses $\frac{n^2}{2(m-1)}$ handles, while the lower bound is $\frac{n(n-1)}{m(m-1)}$ handles. Therefore, this method is worse than the theoretical minimum by a factor of roughly $m/2$.

\begin{example}
\label{m-2-ex}
Let $n=9$ and $m=4$. Call our dimensions $ABCDEFGHI$, where $ABC$ is the first group and $DEFGHI$ the second. For Case~$m-1$ we use the handles $ABCD$, $ABCE$, $ABCF$, $ABCG$, $ABCH$, and $ABCI$. For Case~0, we cover pairs from $DEFGHI$ optimally, using sets of size four; one such choice is $DEFG$, $DEHI$, and $FGHI$, for a total of nine handles.
\end{example}

\subsection{A Recursive-Doubling Method for Covering 2nd-Order Marginals}
\label{rec-doubling-m-2-subsect}
For a sparse but infinite set of values of $n$, there is a better recursion for $C(n,m,2)$. Use two groups, each with half the dimensions, say $n$ dimensions. You can cover all pairs with one dimension in each group, as follows. Assuming $m$ divides $n$, start with sets consisting of $m/2$ members of one of the groups. We need $2n/m$ such sets for each group. Then, pair the sets for each group in all possible ways, forming $4n^2/m^2$ handles of size $m$. These handles cover all pairs that have one member in each group. To these add the recursively constructed sets of handles for the two groups of size $n$. The implied recurrence for this method is:
$$C(2n,m,2)\le 4n^2/m^2 + 2C(n,m,2)$$
If we use $C(m,m,2)=1$ as the basis, the upper bound on $C(n,m,2)$ implied by this recurrence is
$$C(n,m,2)\le 2n^2/m^2 - 1$$
This bound applies only for those values of $n$ that are $m$ times a power of 2. It does, however, give us an upper bound that is only a factor of 2 (roughly) greater than the lower bound of $\binom{n}{2}/\binom{m}{2}$. Additionally, if we attempt to address values of $n$ that are not $m$ times a power of $2$, by rounding up to the nearest such value, we increase $n$ by a factor that approaches $2$ for large values of $n$.  Doing so increases the replication rate by a factor of at most $4$, so the construction in this section improves on that of the previous section for sufficiently large $m$.

\begin{example}
\label{recursive-doubling-ex}
Let $n=m=4$, and suppose the dimensions are $ABCD$ in the first group and $EFGH$ in the second group. We cover all pairs of these eight dimensions with sets of size four, as follows. We first cover the singletons from $ABCD$ using two sets of size 2, say $AB$ and $CD$. Similarly, we cover all singletons from $EFGH$ using $EF$ and $GH$. Then we pair $AB$ and $CD$ in all possible ways with $EF$ and $GH$, to get $ABEF$, $ABGH$, $CDEF$, and $CDGH$. Finally, add covers for each of the groups. A single handle of size four, $ABCD$, covers all pairs from the first group, and the handle $EFGH$ covers all pairs from the second group, for a total of six handles.
\end{example}

\subsection{The General Case}
\label{general-case-subsect}
Finally, we offer a recurrence for $C(n,m,k)$ that works for all $n$ and for all $m>k$. it does not approach the lower bound, but it is significantly better than using one handle per marginal. This method generalizes that of Section~\ref{2-3-slow-subsect}. We use two groups. The first has $m-k+1$ of the dimensions, say $A_1,A_2,\ldots, A_{m-k+1}$, while the second has the remaining $n-m+k-1$ dimensions. The handles are of two types:

\begin{enumerate}
\item
One group of handles contains $A_1A_2\cdots A_{m-k+1}$, i.e., all of group~1, plus any $k-1$ dimensions from group~2. There are $\binom{n-m+k-1}{k-1}$ of these handles, and each has exactly $m$ members.
\item
The other handles are formed recursively to cover the dimensions of group~2, and have none of the members of group~1. There are $C(n-m+k-1,m,k)$ of these handles.
\end{enumerate}

We claim that every marginal of size $k$ is covered by one of these handles. If the marginal has at least one dimension from group~1, then it has at most $k-1$ from group~2. Therefore in is covered by the handles from (1). And if the marginal has no dimensions from group~1, then it is surely covered by a handle from (2).
As a shorthand, let $C(n)$ stand for $C(n,m,k)$. The recurrence for $C(n)$ implied by this construction is
\begin{equation}
\label{general-eq}
C(n)\le \binom{n-m+k-1}{k-1} + C(n-m+k-1)
\end{equation}
We shall prove that:

\begin{theorem}
\label{general-th}
$C(n)\le \binom{n}{k}/(m-k+1)$ for $n$ equal to 1 plus an integer multiple of $m-k+1$.
\end{theorem}

\begin{proof}
The proof is an induction on $n$.

\noindent{\small BASIS}: We know $C(m)=1$, and $\binom{m}{k}/(m-k+1) \ge 1$ for any $1\le k<m$.

\noindent{\small INDUCTION}: We know from Equation~\ref{general-eq} that
$$\binom{n-m+k-1}{k-1} + \frac{\binom{n-m+k-1}{k}}{(m-k+1)}$$
is an upper bound on $C(n)$. We therefore need to show that
$$\frac{\binom{n}{k}}{m-k+1} \ge \binom{n-m+k-1}{k-1} + \frac{\binom{n-m+k-1}{k}}{m-k+1}$$
Equivalently,
\begin{equation}
\label{general2-eq}
\binom{n}{k} \ge (m-k+1)\binom{n-m+k-1}{k-1}+ \binom{n-m+k-1}{k}
\end{equation}
The left side of Equation~\ref{general2-eq} is all ways to pick $k$ things out of $n$. The right side counts a subset of these ways, specifically those ways that pick either:

\begin{enumerate}
\item
Exactly one of the first $m-k+1$ elements and $k-1$ of the remaining elements, or
\item
None of the first $m-k+1$ elements and $k$ from the remaining elements.
\end{enumerate}
Thus, Equation~\ref{general2-eq} holds, and $C(n,m,k)\le \binom{n}{k}/(m-k+1)$ is proved.
\end{proof}

Theorem~\ref{general-th} applies only for certain $n$ that form a linear progression. However, we can prove similar bounds for $n$ that are not of the form 1 plus an integer multiple of $m-k+1$ by using a different basis case. The only effect the basis has is (possibly) to add a constant to the bound.

The bound of Theorem~\ref{general-th} plus Theorem~\ref{rr-th} gives us an upper bound on the replication rate:

\begin{corollary}
\label{general-upper-corr}
We can compute all $k^{\mathrm{th}}$-order marg\-i\-nals using reducers of size $q=d^m$, for $m>k$, with a replication rate of $r\le \binom{n}{k}/(m-k+1)$.
\end{corollary}

\subsection{Handles of Size 4 Covering Marginals of Size 3}
\label{3-4-subsect}
We can improve on Theorem~\ref{general-th} slightly for the special case of $m=4$ and $k=3$. The latter theorem gives us $C(n,4,3)\le \binom{n}{3}/2$, or approximately $C(n,4,3)\le n^3/12$, but we can get $C(n,4,3)\le n^3/16$ by the following method, at least for a sparse but infinite set of values of $n$. Note that in comparison, the lower bound for $C(n,4,3)$ is approximately $n^3/24$.

To get the better upper bound, we generalize the strategy of Section~\ref{2-3-subsect}. Let the dimensions be placed into four groups, with $n$ dimensions in each group. Assume the members of each group are assigned ``indexes'' 1 through $n$.

\begin{enumerate}
\item
Form $n^3$ handles consisting of those sets of dimensions, one from each group, the sum of whose indexes is a multiple of $n$.
\item
For each of the six pairs of groups, recursively cover the members of those two groups together by a set of $C(2n,4,3)$ handles.
\end{enumerate}

Observe that every triple of dimensions is either from three different groups, in which case it is covered by one of the handles from (1), or it involves members of at most two groups, in which case it is covered by a handle from (2). We conclude that:
$$C(4n,4,3)\le n^3 + 6C(2n,4,3)$$
This recurrence is satisfied by $C(n,4,3) = n^3/16$. If we start with, say, $C(4,4,3)=1$, we can show $n^3/16$ is an upper bound on $C(n,4,3)$ for all $n\ge4$ that is a power of two.

\noindent{\bf Aside}: It appears that this algorithm and that of Section~\ref{2-3-subsect} are {\em not} instances of a more general algorithm. That is, there is no useful extension to $C(n,k+1,k)$ for $k>3$.

\section{Optimal Handles are Subcubes}
\label{isoper-sect}
We shall now demonstrate that for a given reducer size $q$, the largest number of marginals of a given order $k$ that we can cover with a single reducer occurs when the reducer gets all tuples needed for a marginal of some higher order $m$. The proof extends the ideas found in \cite{Bollobas86, HLW06} regarding isoperimetric inequalities for the hypercube. In general, an ``isoperimetric inequality'' is a lower bound on the size of the perimeter of a shape, e.g., the fact that the circle has the smallest perimeter of any shape of a given area. For particular families of graphs, these inequalities are used to show that any set of nodes of a certain size must have a minimum number of edges that connect the set to a node not in the set.

We need to use these inequalities in the opposite way~-- to give upper bounds on the number of edges {\em covered}; i.e., both ends of the edge are in the set. For example, in \cite{ADSU13} the idea was used to show that a set of $q$ nodes of the $n$-dimensional Boolean hypercube could not cover more than $\frac{q}{2}\log_2q$ edges. That upper bound, in turn, was needed to give a lower bound on the replication rate (as a function of $q$, the reducer size) for MapReduce algorithms that solve the problem of finding all pairs of inputs at Hamming distance 1.

Here, we have a similar goal of placing a lower bound on replication rate for the problem of computing the $k^{\mathrm{th}}$-order marginals of a data cube of $n$ dimensions, each dimension having extent $d$, using reducers of size $q$. The necessary subgoal is to put an upper bound on the number of subcubes of $k$ dimensions that can be wholly contained within a set of $q$ points of this hypercube. We shall call this function $f_{k,n}(q)$. Technically, $d$ should be a parameter, but we shall assume a fixed $d$ in what follows. We also note that the function does not actually depend on the dimension $n$ of the data cube.

\subsection{Binomial Coefficients with Noninteger Arg\-u\-ments}
\label{gamma-subsect}
Our bound on the function $f_{k,n}(q)$ requires us to use a function that behaves like the binomial coefficients $\binom{x}{y}$, but is defined for all nonnegative $x$ and $y$, not just for integer values (in particular, $x$ may be noninteger, while $y$ will be an integer in what follows). The needed generalization uses the gamma function \cite{wikipedia-gamma} $\Gamma(t) = \int_0^{\inf} x^{t-1}e^{-x}dx$. When $t$ is an integer, $\Gamma(t) = (t-1)!$. But $\Gamma(t)$ is defined for nonintegral $t$ as well. Integration by parts lets us show that $\Gamma$ always behaves like the factorial of one less than its argument:
\begin{equation}
\label{gamma-eq}
\Gamma(t+1) = t\Gamma(t)
\end{equation}
If we generalize the expression for $\binom{u}{v}$ in terms of factorials from $\frac{u!}{v!(u-v)!}$ to
\begin{equation}
\label{gamma2-eq}
\binom{u}{v} = \frac{\Gamma(u+1)}{\Gamma(v+1)\Gamma(u-v+1)}
\end{equation}
then we maintain the property of binomial coefficients that we need in what follows:

\begin{lemma}
\label{gamma-lemma}
If $\binom{x}{y}$ is defined by the expression of Equation~\ref{gamma2-eq}, then
$$\binom{x}{y} = \binom{x-1}{y} + \binom{x-1}{y-1}$$
\end{lemma}

\begin{proof}
If we use Equation~\ref{gamma2-eq} to replace the binomial coefficients, we get
$$\frac{\Gamma(x+1)}{\Gamma(y+1)\Gamma(x-y+1)} =
\frac{\Gamma(x)}{\Gamma(y+1)\Gamma(x-y)} +
\frac{\Gamma(x)}{\Gamma(y)\Gamma(x-y+1)}$$
The above equality can be proved if we use Equation~\ref{gamma-eq} to replace $\Gamma(x+1)$ by $x\Gamma(x)$, \hbox{$\Gamma(x-y+1)$} by \hbox{$(x-y)\Gamma(x-y)$}, and $\Gamma(y+1)$ by $y\Gamma(y)$.
\end{proof}
In what follows, we shall use $\binom{u}{v}$ with the understanding that it actually stands for the expression given by Equation~\ref{gamma2-eq}.

\subsection{The Upper Bound on Covered Subcubes}
\label{isoper-proof-subsect}
We are now ready to prove the upper bound on the number of subcubes of dimension $k$ that can be covered by a set of $q$ nodes.

\begin{theorem}
\label{isoper-th}
$$f_{k,n}(q)\le \frac{q}{d^k}\binom{\log_dq}{k}$$
\end{theorem}

\begin{proof}
The proof is a double induction, with an outer induction on $k$ and the inner induction on $n$.
\noindent{\small BASIS}: The basis is $k=0$. The ``0th-order'' marginals are single points of the data cube, and the theorem asserts that $f_{0,n}(q)\le q$. Since $q$ is the largest number of points at a reducer, the basis is holds, independent of $n$.
\noindent{\small INDUCTION}: We assume the theorem holds for smaller values of $k$ and all $n$, and also that it holds for the same value of $k$ and smaller values of $n$. Partition the cube into $d$ subcubes of dimension $n-1$, based on the value in the first dimension. Call these subcubes the {\em slices}. The inductive hypothesis applies to each slice. Suppose that the $i$th slice has $x_i$ of the $q$ points. Note $\sum_{i=1}^d x_i = q$. There are two ways a $k$-dimensional subcube can be covered by the original $q$ points:

\begin{enumerate}
\item
The subcube of dimension $k$ has a fixed value in dimension 1, and it is contained in one of the $d$ slices.
\item
Dimension 1 is one of the $k$ dimensions of the subcube, so the subcube has a $(k-1)$-dimensional projection in each of the slices.
\end{enumerate}

Case~(1) is easy. By the inductive hypothesis, there can be no more than
$$\sum_{i=1}^d \frac{x_i}{d^k}\binom{\log_dx_i}{k}$$
subcubes of this type covered by the $q$ nodes.
For Case~(2), observe that the number of $k$-dimensional subcubes covered can be no larger than the number of subcubes of dimension $k-1$ that are covered by the smallest of the $d$ slices. The inductive hypothesis also applies to give us an upper bound on these numbers. Therefore, we have an upper bound on $f_{k,n}(q)$:
\begin{equation}
\label{fkn-eq}
f_{k,n}(q)\le \sum_{i=1}^d \frac{x_i}{d^k}\binom{\log_dx_i}{k} + \min_i \frac{x_i}{d^{k-1}}\binom{\log_dx_i}{k-1}
\end{equation}
We claim that Equation~\ref{fkn-eq} attains its maximum value when all the $x_i$'s are equal. We can formally prove this claim by studying the derivatives of this function, however for brevity we will only give an informal proof of this claim.

Suppose that were not true, and the largest value of the right side, subject to the constraint that $\sum_{i=1}^dx_i=q$, occurred with unequal $x_i$'s. We could add $\epsilon$ to each of those $x_i$'s that had the smallest value, and subtract small amounts from the larger $x_i$'s to maintain the constraint that the sum of the $x_i$'s is $q$. The result of this change is to increase the minimum in the second term on the right of Equation~\ref{fkn-eq} at least linearly in $\epsilon$. However, since any power of $\log x_i$ grows more slowly than linearly in $x_i$, there is a negligible effect on the first term on the right of Equation~\ref{fkn-eq}, since the sum of the $x_i$'s does not change, and redistributing small amounts among logarithms will have an effect less than the amount that is redistributed.

Now, let us substitute $x_i = q/d$ for all $x_i$ in Equation~\ref{fkn-eq}. That change gives us a true upper bound on $f_{k,n}(q)$ which is:
$$f_{k,n}(q)\le \frac{q}{d^k}\left[\binom{\log_dq - 1}{k} + \binom{\log_dq - 1}{k-1}\right]$$
But Lemma~\ref{gamma-lemma} tells us $\binom{x}{y} = \binom{x-1}{y} + \binom{x-1}{y-1}$, so we can conclude the theorem when we let $x=\log_dq$ and $y=k$.
\end{proof}

We can now apply Theorem~\ref{isoper-th} to show that when $q$ is the size we need to hold all tuples of the data cube that belong to an $m$th-order marginal for some $m>k$, then the number of $k^{\mathrm{th}}$-order marginals covered by this reducer is maximized if we send it all the tuples belonging to a marginal of order $m$.

\begin{corollary}
\label{isoper-corr}
If $q=d^m$ for some $m>k$, then no selection of $q$ tuples for a reducer can cover more $k^{\mathrm{th}}$-order marginals than choosing all the tuples belonging to an $m$th-order marginal.
\end{corollary}

\begin{proof}
When $q=d^m$, the formula of Theorem~\ref{isoper-th} becomes $f_{k,n}(q) = d^{m-k}\binom{m}{k}$. That is exactly the number of marginals of order $k$ covered by a marginal of order $m$. To observe why, note that we can choose to fix any $m-k$ of the $m$ dimensions that are not fixed in the $m$th-order marginal. We can thus choose $\binom{m}{m-k}$ sets of dimensions to fix, and this value is the same as $\binom{m}{k}$. We can fix the $m-k$ dimensions in any of $d^{m-k}$ ways, thus enabling us to cover $d^{m-k}\binom{m}{k}$ marginals of order $k$.
\end{proof}

\subsection{The Lower Bound on Replication Rate}
\label{general-lower-subsect}
An important consequence of Theorem~\ref{isoper-th} is that we can use our observations about handles and their covers to get a lower bound on replication rate.

\begin{corollary}
\label{general-lower-corr}
If we compute all $k^{\mathrm{th}}$-order marg\-i\-nals using reducers of size $q$, then the replication rate must be at least $r\ge \binom{n}{k}/\binom{\log_dq}{k}$.
\end{corollary}

\begin{proof}
Suppose we use some collection of reducers, where the $i$th reducer receives $q_i$ inputs. There are $d^{n-k}\binom{n}{k}$ marg\-i\-nals that must be computed. By Theorem~\ref{isoper-th}, we know that a reducer with $q_i$ inputs can compute no more than $\frac{q_i}{d^k}\binom{\log_dq_i}{k}$ marginals of order $k$, so
\begin{equation}
\label{general-lower-eq}
d^{n-k}\binom{n}{k} \le \sum_i \frac{q_i}{d^k}\binom{\log_dq_i}{k}
\end{equation}
If we replace the occurrences of $q_i$ in the expression $\log_dq_i$ by $q$ (but leave them as $q_i$ elsewhere), we know the right side of Equation~\ref{general-lower-eq} is only increased. Thus, Equation~\ref{general-lower-eq} implies:
$$d^{n-k}\binom{n}{k} \le \frac{\binom{\log_dq}{k}}{d^k}\sum_iq_i$$
We can further rewrite as:
$$\frac{\sum_iq_i}{d^n} \ge \frac{\binom{n}{k}}{\binom{\log_dq}{k}}$$
The left side is in fact the replication rate, since it is the sum of the number of inputs received by all the reducers divided by the number of inputs. That observation proves the corollary.
\end{proof}

In the case $q=d^m$, Corollary~\ref{general-lower-corr} becomes $r\ge \binom{n}{k}/\binom{m}{k}$. In general, Corollary~\ref{general-lower-corr} says that the replication rate grows rather slowly with $q$. Multiplying $q$ by $d$ (or equivalently, adding 1 to $m$) has the effect of multiplying $r$ by a factor $\binom{m+1}{k}/\binom{m}{k} = (m+1)/(m+1-k)$, which approaches 1 as $m$ gets large.

\section{Dimensions With Different Sizes}
\label{different-sizes-sect}
Let us now take up the case of nonuniform extents for the dimensions. Suppose that the $i$th dimension has $d_i$ different values. Our first observation is that whether you focus on the lower bound on replication rate of Corollary~\ref{general-lower-corr} or the upper bound of Corollary~\ref{general-upper-corr}, the replication rate is a slowly growing function of the reducer size. Thus, if the $d_i$'s are not wildly different, we can take $d$ to be $\max_id_i$. If we select handles based on that assumption, many of the reducers will get fewer than $d^m$ inputs. But the replication rate will not be too different from what it would have been had, say, all reducers been able to take the average number of inputs, rather than the maximum.

\subsection{The General Optimization Problem}
\label{varying-subsect}
We can reformulate the problem of covering sets of dimensions that represent marginals by larger sets that represent handles as a problem with weights. Let the {\em weight} of the $i$th dimension be $w_i = \log d_i$. If $q$ is the reducer size, then we can choose a handle to correspond to a marginal that aggregates over any set of dimensions, say $D_{i_1},D_{i_2},\dots,D_{i_m}$, as long as
\begin{equation}
\label{varying-eq}
\sum_{j=1}^m w_{i_j} \le \log q
\end{equation}

Selecting a smallest set of handles that cover all marginals of size $k$ and satisfy Equation~\ref{varying-eq} is surely an intractable problem. However, there are many heuristics that could be used. An obvious choice is a greedy algorithm. We select handles in turn, at each step selecting the handle that covers the most previously uncovered marginals.

\subsection{Generalizing Fixed-Weight Methods}
\label{generalize-d-subsect}
Each of the methods we have proposed for selecting handles assuming a fixed $d$ can be generalized to allow dimensions to vary. The key idea is that each method involves dividing the dimensions into several groups. We can choose to assign dimensions to groups according to their weights, so all the weights within each group are similar. We can then use the maximum weight within a group as the value of $d$ for that group. If done correctly, that method lets us use larger handles to cover the group(s) with the smallest weights, although we still have some unused reducer capacity typically.

We shall consider one algorithm: the method described in Section~\ref{2-3-subsect} for covering second-order marginals by third-order handles. Recall this algorithm divides $3n$ dimensions into three groups of $n$ dimensions each. We can take the first group to have the smallest $n$ weights, the third group to have the largest weights, and the second group to have the weights in the middle. We then take the weight of a group to be the maximum of the weights of its members. We choose $q$ to be 2 raised to the power that is the sum of the weights of the groups. Then just as in Section~\ref{2-3-subsect} we can cover all marginals that include one dimension from two different groups by selecting $n^2$ particular handles, each of which has a member from each group.

We complete the construction by recursively covering the pairs from a single group. The new element is that the way we handle a single group depends on its weight in relation to $\log q$. The effective value of $m$ (the order of the marginals used as handles) may not be 3; it could be any number. Therefore, we may have to use another algorithm for the individual groups. We hope that an example will make the idea clear.

\begin{example}
\label{varying-ex}
Suppose we have 12 dimensions, four of which have extent up to 8 (weight 3), four of which have extent between 9 and 16 (weight 4), and four of which have extent between 17 and 64 (weight 6). We thus divide the dimensions into groups of size 4, with weights 3, 4, and 6, respectively. The appropriate reducer size is then $q=2^{3+4+6} = 2^{13} = 8192$. We choose 16 handles of size three to cover the pairs of dimensions that are not from the same group.
Now, consider the group of four dimensions with extent 8 (weight 3). With reducers of size 8192 we can accommodate marginals of order 4; in fact we need only half that reducer size to do so. Thus, a single handle consisting of all four dimensions in the group suffices.

Next, consider the group with extent 16 and weight 4. Here we can only accommodate a third-order marginal at a reducer of size 8192, so we have to use three handles of size three to cover any two of the four dimensions in this group. And for the last group, with extent 64 and weight 6, we can only accommodate a second-order marginal at a reducer, and therefore we need six handles, each of which is one of the $\binom{4}{2}$ pairs of dimensions in the last group. We therefore cover all pairs of the 12 dimensions with $16+1+3+6 = 26$ handles.
\end{example}

\section{Conclusions and Open Problems}
\label{conclusions-sect}
Our goal was to minimize the communication (``replication rate'') for MapReduce computations of the marginals of a data cube. We showed how strategies for assigning work to reducers so that each reducer can compute a large number of marginals of fixed order can be viewed as the problem of ``covering'' sets of a fixed size (``marginals'') by a small number of larger sets than contain them (``handles''). We have offered lower bounds and several recursive constructions for selecting a set of handles. Except in one case, Section~\ref{2-3-subsect}, there is a gap between the lower and upper bounds on how many handles we need. We believe there are many opportunities for finding better constructions of handles.

A second important contribution was the proof that our view of the problem is valid. That is, we showed that the strategy of giving each reducer the inputs necessary to compute one marginal of higher order maximized the number of marginals a reducer could compute, given a fixed bound on the number of inputs a reducer could receive. However, this result was predicated on there being the same size extent for each dimension of the data cube. While we offer some modifications to the proposed algorithms for the case where the extents differ in size, there is no proof that an approach where each reducer is assigned the inputs for a higher-order marginal will be best.

Part of the problem is that when the dimensions have different extents, the marginals require different numbers of inputs. Therefore, if we choose to assign one higher-order marginal to a reducer, and that marginal aggregates over many dimensions with small extent, this reducer can cover many marginals with a relatively small number of inputs. But if we want to compute all marginals of a fixed order, we must also compute the marginals that aggregate over dimensions with large extents. If the number of inputs a reducer can receive is fixed, then those marginals must be computed by reducers that cover relatively few marginals. Thus, an upper bound on the number of marginals that can be covered by a reducer of fixed size will be unrealistic, and not attainable by all the reducers used in a single MapReduce algorithm.


\end{document}